\documentclass[12pt]{article}
\usepackage{graphicx,amsmath} 

\usepackage{times,graphicx,parskip,natbib,setspace}
\usepackage{amsmath,tikz}
\usepackage{amssymb}
\usepackage{accents}
\usepackage{float}
\usepackage{subfigure}
\usepackage{url,color}
\usepackage{booktabs}
\usetikzlibrary{arrows.meta, calc, decorations.markings}

\newtheorem{theorem}{Result}[section]

\title{Principal nested spheres for high-dimensional data}
\author{Mymuna Monem$^{1,3}$, Ian L.~Dryden$^2$ and Florence George$^1$\\ $\;$ \\
$^1$Department of Mathematics and Statistics, Florida International University\\
$^2$School of Mathematical Sciences, University of Nottingham\\
$^3$Baptist Health South Florida, Coral Gables, Florida.
}
\date{ }

\begin{document}
\maketitle

\onehalfspacing

\begin{abstract}
The method of Principal Nested Spheres (PNS) is a non-linear dimension reduction technique for spherical data. The method is a backwards fitting procedure, starting with fitting a high-dimensional sphere and then successively reducing dimension at each stage. After reviewing the PNS method in detail, we introduce some new methods for model selection at each stage between great and small subspheres, based on 
the Kolmogorov-Smirnov test, a variance test and a likelihood ratio test. The current PNS fitting method is slow for high-dimensional spherical data, and so we introduce a fast PNS method which involves an initial principal components analysis decomposition to select a basis for lower dimensional PNS. A new visual method called the PNS biplot is introduced for examining the effects of the original variables on the PNS, and this involves procedures 
for back-fitting from the PNS scores back to the original variables. 
The methodology is illustrated with two high-dimensional datasets from cancer research: Melanoma proteomics data with 500 variables and 205 patients, and a Pan Cancer dataset with 12,478 genes and 300 patients. In both applications the PNS biplot is used to select variables for effective classification.  
\end{abstract}


\section{Introduction}

Principal Nested Spheres (PNS)  is a 
non-linear dimension reduction method for 
spherical data, and was introduced by \citet{Jungetal12}. The method involves successively fitting a sequence of nested subspheres to a dataset on a sphere, where one aims to retain as much information as possible at each stage. The method is carried out in a backwards manner, fitting subspheres at the highest dimension and then conditional on that fit successively fitting the next 
lower dimensional sphere, and so on. This is analogous to backwards variable selection in regression, where one fits a full model and then successively removes predictors at each stage \citep{Marronetal10}.

PNS has been used in several contexts since \citet{Jungetal12}
introduced the work, with initial applications to spheres and planar landmark shapes. \citet{Drydenetal19} extended the method to 3D shape spaces, with applications to peptide shape analysis; \citet{pennec2018barycentric} proposed the nested approach of barycentric subspace analysis on manifolds; 
\citet{Zoubouloglou2023_STPCA} developed methodology for the analysis of data on the torus, using a mapping to the sphere and then applying PNS; and 
\citet{YangVemuri2020_NestedGrassmannians} consider a similar idea to PNS with nested Grassmanian manifolds.

In this paper we introduce several extensions of the PNS method, including fast PNS for high-dimensional spheres, the PNS biplot for interpretation,  and new approaches to model choice. 
We apply the methodology to a study of Melanoma proteomics data \citep{Browetal10}, where it is of interest to distinguish between Stage 1 and Stage 4 cancer, and 
we analyse some Pan Cancer RNA-seq data \citep{Hoadley2014TCGA-PanCan}, where it is of interest to compare six different types of cancer.

\section{Principal nested spheres}
In this section we review the method of PNS in detail, in particular focusing on the construction, parameterization, PNS scores and estimation. 
\subsection{PNS construction}
Consider the sphere $S^d$, the 
$d$-dimensional sphere in $(d+1)$ real dimensions $\mathbb{R}^{d+1}$. The method of principal nested spheres \citep{Jungetal12} involves finding a sequence of nested spheres.
At the first level of the sequence we  find a {subsphere} $A_{d-1}$ of $S^d$, defined by an orthogonal axis $v_1\in S^d$ and an angle $r_1\in(0,\pi/2]$, given by 
$$A_{d-1}(v_1,r_1)=\{x \in S^d:\rho_d(v_1,x)=r_1\},$$
where $\rho_d(x,y)$ is the spherical distance between any two points $x$ and $y$ on $S^d$ which is also given by $\rho_d(x,y)=\cos^{-1} (x^Ty)$. If the angle is $r=\pi/2$ then $A_{d-1}$ is a { great subsphere}, and if $r < \pi/2$ then $A_{d-1}$ is a { small subsphere}. 

The subsphere $A_{d-1}(v_1,r_1)$ is then translated, rotated and rescaled 
to a unit sphere $S^{d-1}$ by the 1-1 and onto transformation $S^{d-1} = f_1(A_{d-1})$, with  $f_1:\mathbb{R}^d\rightarrow \mathbb{R}^d$. The inverse transformation is $A_{d-1} = f_1^{-1}(S^{d-1})$. 

At the next stage, we find a subsphere on $S^{d-1}=f_1(A_{d-1})$ which is denoted by $A_{d-2}(v_2,r_2)$. This is then translated, rotated and rescaled to $S^{d-2}$ using the transformation \begin{align*}
S^{d-2} &= f_2(  A_{d-2} ) = f_2 \circ  f_1(A_{d-1}).
\end{align*}  
So the construction of subspheres proceeds in a backward fashion $A_{d-1},A_{d-2},\ldots,A_1$ with each step transforming to the sphere $S^{d-1},\ldots,S^2,S^1$. To write mathematically,
\begin{align*}
S^{d-1} &= f_1(A_{d-1}) \\
S^{d-2} &= f_2(A_{d-2}) = f_2 \circ f_1(A_{d-1}) \\
&\vdots \\
S^1 &= f_{d-1} \circ f_{d-2} \circ \ldots \circ f_2 \circ f_1(A_{d-1})
\end{align*}
In the final step on the sphere $S^1$, we find a 0-dimensional subsphere $A_0$ which is a point on $S^1$ and is written as $A_0 \equiv v_d$. 

In order to calculate the { nested subspheres}  $\zeta_i$ of the original sphere we need to apply recursively the inverse transformation $f_{k}^{-1}$. The resulting decomposition of spheres sequentially gives $k$-dimensional nested subspheres denoted by $\zeta_k$ of the data for each $k=0,1,2,...,d-1$. This $\zeta_k$ is a sequence of nested subspheres of $S^d$ can be expressed as
\begin{equation}\label{rel}   \zeta_0\subset\zeta_1\subset\zeta_2\subset...\subset\zeta_{d-1}\subset\ S^d.
\end{equation}

The top nested subsphere is $\zeta_{d-1} = A_{d-1}$,
then the remainder of the nested subsphere sequence is
\begin{align*}
\zeta_{d-2} &= f_1^{-1}(A_{d-2}) \\
\zeta_{d-3} &= f_1^{-1} \circ f_2^{-1}(A_{d-3}) \\
&\vdots \\
\zeta_{1} &= f_1^{-1} \circ f_2^{-1} \circ \ldots \circ f_{d-2}^{-1}(A_{1}) \\
\zeta_{0} &= f_1^{-1} \circ f_2^{-1} \circ \ldots \circ f_{d-1}^{-1}(A_{0})
\end{align*}

The final nested sphere $\zeta_0$ is actually a point on $S^{d}$, and this will be referred to as the { PNS mean}. 

To illustrate the nested subspheres and transformations we refer to Figure \ref{subsphere}, which follows from \citet{Jungetal12}.

\begin{figure}[H]
         \centering
\begin{tikzpicture}[scale=0.9, line join=round, line cap=round]

\shade[ball color=gray!20, opacity=0.9] (0,0) circle (2);
\draw[thick] (0,0) circle (2);
\node at (0,-2.6) {$S^d \subset \mathbb{R}^{d+1}$};

\draw[thick, opacity=0.6] (-2,0) arc[start angle=180,end angle=360,x radius=2,y radius=0.5];
\draw[dashed, opacity=0.3] (2,0) arc[start angle=0,end angle=180,x radius=2,y radius=0.5];

\draw[very thick, red, dashed] (0,1.2) ellipse[x radius=1.55, y radius=0.4];
\node[red!70!black] at (2.8,1.3) {$A_{d-1}(v_1, r_1)$};

\draw[->, thick] (0,0) -- (0,2.3);
\node at (0.3,2.3) {$v_1$};

\draw[dashed, thick] (0,0) -- (0,1.2);
\node at (0.25,0.45) {$r_1$};

  \draw[-, thick] (0,0) -- (1.6,1.2);

\draw[thick, blue, -, >=Stealth] (0,0.3) arc[start angle=90, end angle=40, radius=0.3];

\draw[->, thick] (2.2,1.45) .. controls (3.5,2.2) .. (4.9,1.6) node[midway, above] {$f_1$};

\begin{scope}[shift={(6,0)}]
  \shade[ball color=gray!20, opacity=0.9] (0,0) circle (1.6);
  \draw[thick] (0,0) circle (1.6);
  \node at (0,-2.1) {$S^{d-1} \subset \mathbb{R}^{d}$};

  \draw[thick, opacity=0.6] (-1.6,0) arc[start angle=180,end angle=360,x radius=1.6,y radius=0.4];
  \draw[dashed, opacity=0.3] (1.6,0) arc[start angle=0,end angle=180,x radius=1.6,y radius=0.4];

  \draw[very thick, red, dashed] (0,0.9) ellipse[x radius=1.26, y radius=0.3];
  \node[red!70!black] at (2.45,1.0) {$A_{d-2}(v_2, r_2)$};

  \draw[->, thick] (0,0) -- (0,1.9);
  \node at (0.3,1.9) {$v_2$};

  \draw[dashed, thick] (0,0) -- (0,0.9);
  \node at (0.25,0.4) {$r_2$};

  \draw[-, thick] (0,0) -- (1.3,0.9);

  \draw[thick, blue, -, >=Stealth] (0,0.3) arc[start angle=90, end angle=40, radius=0.3];

\end{scope}

\draw[->, thick] (7.8,1.2) .. controls (8.9,1.8) .. (10.5,1.25) node[midway, above] {$f_2$};

\begin{scope}[shift={(11.5,0)}]
  \shade[ball color=gray!20, opacity=0.9] (0,0) circle (1.3);
  \draw[thick] (0,0) circle (1.3);
  \node at (0,-1.8) {$S^{d-2} \subset \mathbb{R}^{d-1}$};

  \draw[thick, opacity=0.6] (-1.3,0) arc[start angle=180,end angle=360,x radius=1.3,y radius=0.35];
  \draw[dashed, opacity=0.3] (1.3,0) arc[start angle=0,end angle=180,x radius=1.3,y radius=0.35];

  \draw[very thick, red, dashed] (0,0.8) ellipse[x radius=1, y radius=0.25];
  \node[red!70!black] at (2.2,0.9) {$A_{d-3}(v_3, r_3)$};

  \draw[->, thick] (0,0) -- (0,1.6);
  \node at (0.3,1.6) {$v_3$};

  \draw[dashed, thick] (0,0) -- (0,0.8);
  \node at (0.25,0.4) {$r_3$};

  \draw[-, thick] (0,0) -- (1.0,0.8);

  \draw[thick, blue, -, >=Stealth] (0,0.25) arc[start angle=90, end angle=40, radius=0.25];
\end{scope}

\draw[->, thick] (13.2,1.2) .. controls (13.9,1.8) .. (15.2,1.3) node[midway, above] {$f_3$};

\end{tikzpicture}
   \caption{The subsphere $A_{d-1}(v_1,r_1)$ of $S^d$ (left) and its transformation by $f_1$ to give $S^{d-1}$ (middle). The subsequent sequence relation gives subsphere $A_{d-2}(v_2,r_2)$ (middle) and its transformation by $f_2$ to $S^{d-2}$ (right). Then subsphere $A_{d-3}(v_3,r_3)$ is obtained and
   transformed by $f_3$ and the sequence continues further until $A_0$ is obtained on $S^1$. Each sphere is one dimension lower than the previous sphere. }
     \label{subsphere}
\end{figure}
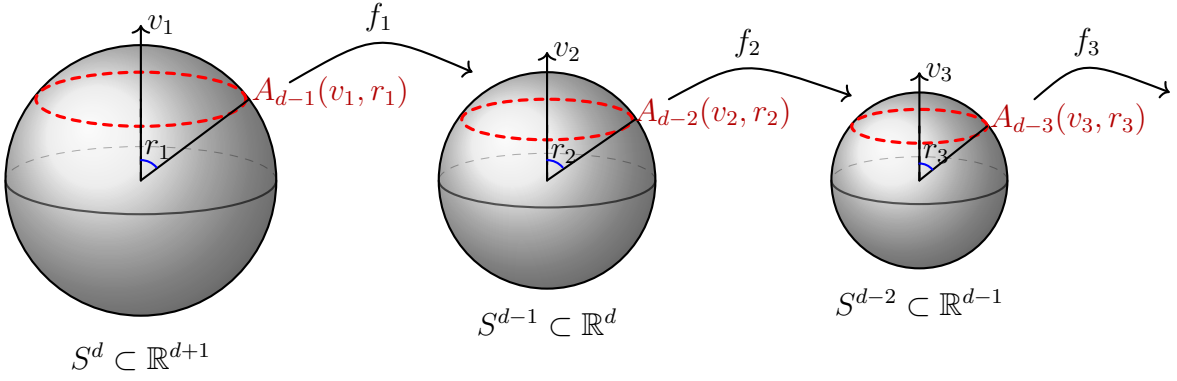

\subsection{The PNS parameter space}
The PNS parameter space consists of the parameters $\{ (v_k,r_k), k=1,\ldots,d-1, \; \; {\rm and} \; \;  v_d  \}$, where $v_k$ is an orthogonal axis $(d+2-k)$-vector, and $r_k$ is the angle between $v_k$ and the intersection between each sphere and subsphere, that is $r_k \in (0,\frac{\pi}{2}], \; k=1,\ldots,d-1$.

\begin{theorem} The dimension of the PNS parameter space is 
\begin{equation}\label{D}
  D = d(d+3)/2 - 1. 
\end{equation}
\end{theorem}

\begin{proof}
As each $v_k$ is on a sphere, we have $\|v_k\| = 1$, and so the dimension of each $v_k$ is $d+1-k$.  So, the number of dimensions of the PNS parameter space is 
$d + (d-1) + \ldots + 2 + 1$ for the vectors $v_k, k=1,\ldots,d$, plus $d-1$ for the radii/angles. Hence, the total number of parameters is $d(d+1)/2 + (d-1) = d(d+3)/2 - 1$. 
\end{proof}

Let us write $\Phi$ for the set of PNS parameters, where 
\begin{equation} \Phi = \{ v_1,\ldots,v_d,r_1,\ldots,r_{d-1} \} .
\label{pars}
\end{equation}
It will also be of interest to consider the 
PNS mean $\zeta_0 = (\phi_1,\ldots,\phi_{d+1})^T \in S^d$ as a parameter vector of interest. The PNS mean is uniquely determined by the PNS parameters $\Phi$ using the transformation 
$\zeta_{0} = f_1^{-1} \circ f_2^{-1} \circ \ldots \circ f_{d-1}^{-1}(A_{0})$. 

\subsubsection*{Special cases}
\begin{enumerate}
    \item Let $d=2$. 
    The parameters are $v_1 \in S^2$ (2 dimensions), $v_2 \in S^1$ (1 dimension) 
    and $r_1$ (1 dimension). So the total number of free parameters is $4 = 2+1+1$
    which is equal to $D= d(d+3)/2-1 = 2\cdot 5/2 -1 = 4$.

    \item Let $d=3$.
    The parameters are 
    $v_1 \in S^3$ (3 dimensions), $v_2 \in S^2$ (2 dimension), $v_3 \in S^1$ (1 dimension) 
    and $r_1, r_2$ (2 dimensions). So the total number of free parameters is $8 = 3+2+1+2$ 
    which is equal to $D= d(d+3)/2-1 = 3\cdot 6/2 -1 = 8$.
\end{enumerate}

\subsection{Residuals and PNS scores}
\subsubsection{PNS scores for a single new observation}
Let us first consider one single new observation $x$ on the sphere $S^d$ and we describe the calculation of the residuals and PNS scores given the values of the PNS parameters $\Phi$.
\begin{enumerate}
\item[Step 1:]
 As per our above discussion, the point $x$ will be projected onto one less dimensional subsphere i.e $A_{d-1}$ so we can write  $x^{(d)*} = P ( x^{(d)} ) \in A_{d-1}(v_1,r_1) \in S^{d}$, where $P$ is the projection function which finds the closest point $x^{(d)*}$ on the subsphere $A_{d-1}$, and we can express this by 
$$x^{(d)*} = P( x^{(d)} )  = {\rm arg} \inf_y  \rho_d( x^{(d)} , y ) . $$ 

To find the PNS score we need to calculate the residual $\xi_d$ of $x$ from a subsphere $A_{d-1}(v_1,r_1)$ of $S^d$ which is defined as the signed length of the minimal geodesic distance between the subsphere to the data point $x$ expressed by $\xi_d =\rho_d(x,v_1)-r_1= \cos^{-1}( v_1^Tx^{(d)} ) - r_1. $
Note that  $ | \xi_{d} | =  \rho_d(x^{(d)} , x^{(d)*} )$, which is the distance between the original point $x^{(d)}$ and projected point $x^{(d)*}$.  The signed residual $\xi_{d}$ will be used to calculate a PNS score. The sign of $\xi_d$ depends on whether $x$ is in the interior or exterior of the geodesic subsphere $A_{d-1}$, then it will be negative or positive respectively. 

Now after finding the projected point $x^{(d)*}$, it is then transformed to a point on a new unit sphere $x^{(d-1)} = f_1( x^{(d)*} ) \in S^{d-1}$ using the rotation and rescaling function $f_1()$. 

\item[Step 2:]
Next this point $x^{(d-1)}$  is projected to one lower dimensional subsphere $x^{(d-1)*} = P ( x^{(d-1)} ) \in A_{d-2}(v_2,r_2) \in S^{d-1}$.  The signed residual will be $$ \xi_{d-1}  = \cos^{-1}( v_2^Tx^{(d-1)} ) - r_2 , $$ 
and we have $|\xi_{d-1}| =  \rho_{d-1}(x^{(d-1)} , x^{(d-1)*} )$. The projected point $x^{(d-1)*}$ is then transformed to a point on a new unit sphere $x^{(d-2)} = f_2( x^{(d-1)*} ) \in S^{d-2}$ using the rotation and rescaling function $f_2()$. 

\item[Step $k+1$:]
This process continues and so at an intermediate step $k+1$ the projected point is transformed to $x^{(d-k)} = f_{k}( x^{(d-k+1)*} ) \in S^{d-k}$. 
In turn this point is projected to $x^{(d-k)*} = P ( x^{(d-k)} ) \in A_{d-k-1}(v_{k+1},r_{k+1})$.  The signed residual 
is 
\begin{equation} \xi_{d-k}  = \cos^{-1}( v_{k+1}^Tx^{(d-k)} ) - r_{k+1} . 
\label{PNSresiduals0} \end{equation}
The projected point $x^{(d-k)*}$ is then transformed to a point on a new unit sphere $x^{(d-k-1)} = f_{k+1}( x^{(d-k)*} ) \in S^{d-k-1}$ using the rotation and rescaling function $f_{k+1}()$. 

\item[Step $d$:]
In the final step $d$, the projected point is transformed to $x^{(1)} = f_{d-1}( x^{(2)*} ) \in S^{1}$. The parameter at this final level is $v_{d} \in S^1$ which is an angle, and the residual 
distance $| \xi_1 | = \rho_{1}(x^{(1)}, v_d )$, and 
the sign of the score is determined by whether the point is clockwise or anti-clockwise from $v_d$. Note   $\xi_{1} \in (-\pi,\pi]$.  
\end{enumerate}

Now that the residuals have been obtained,
the full set of PNS scores is then computed by multiplying the $d$-vector of residuals $(\xi_1, \ldots, \xi_{d})$ by functions of the radii to rescale 
the scores to have the correct scale 
of the original high-dimensional sphere $S^d$. 

In particular, the PNS scores are 
$$s_{k} =  a_k  \xi_{k}  \; , \qquad \;  a_k = \prod_{i=1}^{d-k}  \sin(r_i) \;  , \qquad k=1,\ldots,d-1 . 
$$
$s_d = \xi_{d}$, and $a_k$ is the radius of the 
subsphere $\zeta_k$, $k=1,\ldots,k$. We call $s_1,\ldots,s_d$ the {PNS1, PNS2, \ldots, PNS$d$ scores}, respectively. 
To find the range of the PNS scores, we have $s_1 \in  S^1(a_1)$, 
a circle of radius $a_1$. Hence, PNS1 is a circular variable.
The other PNS scores are non-circular 
variables on an interval of $\mathbb{R}$. For great subspheres
the interval is $(-a_k \pi/2, a_k \pi/2)$, and for small subspheres the interval is determined by the maximum distances 
in the interior and exterior of each subsphere.
So, the PNS score space is analogous to 
a type of cylinder, with PNS1 being a circular variable, and PNS2 to 
PNS$d$ being defined on intervals. The PNS score space is a subset of the 
product space  $S^1(a_1) \times \mathbb{R}^{d-1}$. The score $s_1$ comes from the final step (low dimensional) sphere and the subsequent  scores $s_2, s_3\ldots,s_d$ come from successively 
higher dimensional spheres, and we write 
$ \mathbf{s} = (s_1,\ldots,s_d)^T$.
The nonlinear function that maps from the original spherical data $x^{(d)} \in S^d$ to the PNS scores for a particular set of parameters $\Phi$ is written as
\begin{equation}
    \mathbf{s} = h( x^{(d)} | \Phi) , \; \quad h: S^d \to S^1(a_1) \times {\mathbb{R}}^{d-1} \label{gfun}
    \end{equation}
which is a 1-1 and onto map with inverse $x^{(d)} = h^{-1}(\mathbf{s}| \Phi)$. 

\subsubsection{Population model}
A population model for PNS can be constructed from either the 
residuals or the scores given the 
population parameters $\Phi$. For example, the 
residual $\xi_k$ could have truncated Gaussian distributions with mean zero and positive variance $\sigma_k^2, 
k=1,\ldots,d$. Other symmetric distributions 
could also be used, e.g. the Laplace distribution
which is more accommodating of outliers. 

\subsubsection{PNS scores for a sample}
If we have a new sample size $n$ of observations $x_i^{(d)}, i=1,\ldots,n$ on $S^d$ we would like to compute the PNS scores given the PNS parameters. 
We denote the PNS scores for the $i$th observation as the $d$-vector 
\begin{equation} (s_{1,i}, s_{2,i}, \ldots, s_{d-1,i}, s_{d,i})^T, \; i=1,\ldots,n. \label{PNSscores} \end{equation} 
The corresponding residuals are
written as
\begin{equation}
(\xi_{1,i}, \xi_{2,i}, \ldots, \xi_{d-1,i}, \xi_{d,i})^T, \; i=1,\ldots,n.
\label{PNSresiduals}
\end{equation}

\subsection{Least squares estimation}\label{sec:LS}
The main application of PNS is to fit 
a sequence of nested subspheres to 
a dataset of $n$ points on $S^d$. At each stage we wish to 
fit a nested subsphere which fits as close as
possible to the data in the current 
sphere. A convenient method is to use least squares
in order to minimize the sum of squared residuals 
from the data at the higher level to the projected data on the fitted subsphere at the next (lower) level. As we are
finding the best fit at each level the term `principal' nested spheres is appropriate. 

Consider the 
data $x^{(d)}_i, i=1,\ldots,n$ on the highest dimensional sphere $S^d$. 
At the first stage we find the best fitting subsphere $\hat{A}_{d-1}$ by minimizing the sum of squares of the residuals of $x^{(d)}_i$ that is
\begin{equation}
    \sum_{i=1}^{n} \xi_{d,i}(v_{1},r_{1})^2=\sum_{i=1}^{n} \{\rho_d(x^{(d)}_i,v_{1})-r_{1}\}^2,
\end{equation}
where $v_1\in S^d$ and $r_1\in(0,\pi/2]$ and the residuals are
given in (\ref{PNSresiduals}) using (\ref{PNSresiduals0}). Using the least squares estimation method we can find the best estimate for $\hat v_1$ and $\hat r_1$. For example, minimization can be carried out by the R function {\tt optim} \citep{R-package}.

Then in the next stage, conditional on 
$\hat v_1, \hat r_1$, we project and transform 
the data to the sphere of one fewer dimension, i.e. to $x^{(d-1)}_i \in S^{d-1}, i=1,\ldots,n$.
We then find the best-fitting subsphere $\hat{A}_{d-2} \in S^{d-1}$  by minimizing the sum of squares of the residuals of the projected data in $S^{d-2}$ that is
\begin{equation}
    \sum_{i=1}^{n} \xi_{d-1,i}(v_{2},r_{2})^2=\sum_{i=1}^{n} \{\rho_{d-1}(x^{(d-1)}_i,v_{2})-r_{2}\}^2.
\end{equation}

This sequence continues until the final 
stage where $\hat v_d$ is chosen as the Fr\'echet mean
of the angles $x^{(1)}_i \in S^1, i=1,\ldots,n$, i.e. 
$$ \hat v_d = {\rm arg} \inf_{v} \rho_1( x_i^{(1)},v)^2 . $$

In order to denote the set of 
fitted parameters from a PNS fit to 
a dataset we write  $\hat \Phi$, where the parameters $\Phi$ were defined in 
Equation (\ref{pars}).

An implementation of code for fitting principal nested spheres is in the function {\tt pns} in the {\tt shapes} package in R \citep{DrydenGithub}. 
This code is closely based on the original 
Matlab code of Sungkyu Jung \citep{Jungetal12}.

We can also regard the method of fitting as a sequential conditional maximum likelihood (ML) procedure using Gaussian errors, which can be helpful for distributional 
results and investigating properties. 
At each stage the parameters $(v_k, r_k)$ are obtained using ML estimation conditional on the previous stage estimates also obtained using ML estimation.

\section{Properties and extensions}
\subsection{Consistency}
Consider a population PNS model with true parameters $\Phi$ which are identifiable, for example residual errors with mean zero at each stage and strictly positive variance. 
The usual PNS fitting procedure in Section \ref{sec:LS}
corresponds to using independent 
Gaussian errors. Since the estimation procedure is 
an M-estimation problem at each stage, then using standard assumptions on the residual errors we have consistency of the PNS estimates $ \hat \Phi \overset{p}\to \Phi$
 as $n \to \infty$. In addition under mild assumptions 
 $$ \sqrt{n}( \hat  \Phi - \Phi) \overset{d}{\to} N(0, \Sigma) . $$ Note that the estimation variance is propagated down each of the levels, and so uncertainty in the estimates can be obtained using a bootstrap procedure.

\subsection{Model choice: Great versus small subspheres}
An important practical consideration with PNS is choosing a method to decide between fitting a great subsphere versus a small subsphere at each level. The small subsphere will always fit better with a lower error sum of squares as there is an extra angle parameter, but is it a significantly better model? 
We can consider the hypothesis test at each stage between
\begin{center}
    \textit{$H_{0}$:} The residual distribution is the same under the great and small subsphere fits.  
\end{center}
vs
\begin{center}
    \textit{$H_{1}$:} The residual distribution is different under the great and small subsphere fits.  
\end{center}

Some methods for fitting are 
available in the {\tt pns} function of the 
{\tt shapes} library \citep{DrydenGithub} 
with option {\tt sphere.type=} given below in parentheses. \citet{Jungetal12} consider a sequential bootstrap test
({\tt seq.test}) and the Bayesian Information Criterion (BIC) for Gaussian residuals as possible methods for choosing between the great and small subsphere models at each stage of the fitting procedure. 
Other options are to choose small subspheres at each level ({\tt small}), great subspheres at each level ({\tt great}), or great subspheres for initial stages, and then small subspheres ({\tt bi-sphere}). 
We introduce some additional methods for model choice at each stage.

\subsubsection{Kolmogorov-Smirnov test}\label{KStest}
The two-sample Kolmogorov-Smirnov (KS) test \citep[e.g.,][page 309--314] {Conover1971} can be used to test the equality of the 
distribution of the absolute value of the residuals from fitting the great subsphere or small subsphere at each stage (see {\tt ks.test} in R).
The KS test statistic here is defined as
$$D_n = \sup_{x \in \mathbb{R}^+} \left| \hat F_{great}(x) - \hat F_{small}(x) \right|,$$
and we are testing if $H_0: {F}_{great}={F}_{small}$, where $\hat F_{great}, \hat F_{small}$ are the empirical cumulative distribution functions of the absolute value of the residuals under the great subsphere and small subsphere fits at a particular stage. 
The KS test is expected to not be as powerful as the likelihood ratio tests in Section \ref{LR} and so we may be less likely to reject the great subsphere model at each stage. A bias towards great subsphere fits may actually be desirable in many applications, providing a form of regularization that favors the great subsphere models.

\subsubsection{Variance Test}\label{vartest}
The variance test performs an F test to compare the variances of the residuals 
under the great and small subsphere fits. The test statistic is calculated as:
$F = \frac{s_1^2}{s_2^2}$
where $s_1^2$ and $s_2^2$ are the sample variances of the residuals from 
the great and small model fits, and we reject $H_0$ of equal variances if $F > F_{n_1-1,n_2-1,1-\alpha}$ where $F_{\nu_1,\nu_2,1-\alpha}$ is the $1-\alpha$ quantile of the $F_{\nu_1,\nu_2}$ distribution with $\nu_1,\nu_2$ degrees of freedom.
For this method the great subsphere fit is deemed appropriate if the equality of the population variances of the residuals cannot be rejected. Otherwise, we choose the small subsphere fit. 

\subsubsection{Likelihood ratio tests}\label{LR}
We can carry out the test between $H_0$ and $H_1$ using 
likelihood ratio tests from appropriate parametric distributions. In particular, we use the normal 
distribution for the residuals here. We compute the usual log-likelihood ratio $\Lambda$ and reject $H_0$ if $-2 \log \Lambda > \chi^2_{1,1-\alpha}$, where $\chi^2_{\nu,1-\alpha}$ is the $1-\alpha$ quantile of the chi-squared distribution with 
$\nu$ degrees of freedom. Other choices can of course be considered, for example the Laplace distribution, or the exponential or Weibull distribution for the absolute value of the residuals.

\subsection{Fast approximation of high-dimensional PNS using PCA}\label{fastpns}
Following the approximate method introduced by \citet{Drydenetal19} for nested shape spaces, we consider a fast approximation for high-dimensional PNS by first carrying out principal components analysis (PCA) to construct an approximating lower dimensional sphere. Recall that PCA is a very widely used technique for dimension reduction, where an eigen-decomposition of the covariance matrix leads to the PC 
loading vectors given by the eigenvectors, with the variance of each PC given by the eigenvalues  \citep[e.g., see][]{Jolliffe02}. 
The details of the fast PNS method that we introduce are similar to \citet{Drydenetal19},  but there is no Procrustes matching stage here.


Consider $X_1,\ldots, X_n \in \mathcal{S}^d$. We can approximate PNS by first obtaining the arithmetic  mean $\bar X^A = \frac{1}{n} \sum_{i=1}^n X_i$, which is then normalized on the sphere to $\bar X = \bar X^A /\| \bar X^A \|$, and principal component loading vectors 
$V_1,\ldots,V_p$ which are the eigenvectors corresponding to the largest $p$ eigenvalues of the covariance matrix of the orthogonal tangent coordinates at $\bar X$:  $T_i = X_i -  (\bar X^TX_i) \bar X$, where $i=1,\ldots,n$ and $p < n$. 
  For any $1\leqslant p<  d+1$, the $p$-dimensional unit sphere $\mathcal{S}_{\bar X;V_1,\ldots,V_p}$ in the $(p + 1)$-dimensional linear subspace spanned by $\bar X,V_1,\ldots,V_p$ is a $p$-dimensional subsphere of $\mathcal{S}^d$, where $p \le d$, and the first $p$ eigenvalues are required to be strictly positive. If we choose $p$ such that the first $p$ principal components explain a high proportion of the total variation of the data then, without much loss of information, the use of the projection of $X_i$ to $\mathcal{S}_{\bar X;V_1,\ldots,V_p}$ as the input for the principal nested spheres analysis will reduce the initial dimension of the sphere and improve the speed of computation. If $p=d$ then we would like the resulting 
 PNS calculation to be exact, i.e. to lose no information. 

The projection of $X_i \in \mathcal{S}^d$ to $\mathcal{S}_{\bar X;V_1,\ldots,V_p}$ can be approximated by using the principal component scores, where the principal component score for the $i$th configuration on the $j$th principal component is given by 
$\tilde\lambda_{ij}=\langle X_i,V_j\rangle, i=1,\ldots,n, j=1,\ldots,p$.
For this, let
\[W_i=\rho(\bar X,X_i)\frac{T_i}{\|T_i\|}\in\mathcal T_{\bar X}(\mathcal{S}^d) , \]
where $T_{\bar X}(\mathcal{S}^d)$ denotes the tangent space to the sphere $\mathcal{S}^d$ at $\bar X$. Then, $W_i$ is the image of $X_i$ under the inverse exponential map of $\mathcal{S}^d$ at $\bar X$, and, in particular, $\|W_i\|=\rho(\bar X,X_i)$ and 
$\rho()$ is the great circle (Riemannian) distance on $\mathcal{S}^d$. By writing
\[\lambda_{ij}=\langle W_i,V_j\rangle=\frac{\rho(\bar X,X_i)}{\|T_i\|}\tilde\lambda_{ij},\]
the projection of $W_i$ in the subspace, spanned by $V_1,\ldots,V_p$, of the tangent space $\mathcal T_{\bar X}(\mathcal{S}_d)$ is
\begin{equation}\label{eq:tancoords}
U_i=\sum_{j=1}^p\lambda_{ij}V_j.
\end{equation}
Note that $\|U_i\|^2=\sum_{j=1}^p\lambda_{ij}^2$.
The image $X^*_i$ of $U_i$ under the exponential map back in the sphere is
\begin{eqnarray*}
X_i^*&=&\bar X\cos(\|U_i\|)+\frac{U_i}{\|U_i\|}\sin(\|U_i\|)\\
&=&\cos(\|U_i\|)\,\bar X+\frac{\sin(\|U_i\|)}{\|U_i\|}\sum_{j=1}^p\lambda_{ij}V_j\in\mathcal{S}_{\bar X; V_1,\ldots,V_p}
\end{eqnarray*}
where, if $\|U_i\|=0$, we take $X^*_i=(1,0,\ldots,0)$. When $X_i$ is sufficiently close to $\mathcal{S}_{\bar X;V_1,\ldots,V_p}$, $X^*_i$ gives a good approximation to the projection of $X_i$ to $\mathcal{S}_{\bar X;V_1,\ldots,V_p}$. 

Since $\bar X, V_1,\cdots,V_p$ are orthonormal, they form a basis of the $(p + 1)$-dimensional linear subspace spanned by themselves and, in terms of this new basis, $X^*_i$ can be expressed with co-ordinates
 \[ \left(\cos(\|U_i\|),\frac{\sin(\|U_i\|)}{\|U_i\|}\lambda_{i1},\cdots,\frac{\sin(\|U_i\|)}{\|U_i\|}\lambda_{ip}\right).\]  
Note that many of the orthogonal axes in the reduced dimension PNS are expected to line up with the PC axes, especially for the smaller variability components. So, the reduced dimension PNS axes in this representation are often 0's and one element either 1 or very close to 1. 

Given a reduced dimension PNS object and the scores we can map back to obtain the corresponding coordinates on the sphere $\mathcal{S}_{\bar X;V_1,\ldots,V_p}$ using the R command {\tt PNSe2s} from the {\tt shapes} library. 
We will denotes these spherical coordinates as $G = (G_1, G_2, \ldots, G_{p + 1})^\top \in\mathcal{S}_{\bar X;V_1,\ldots,V_p}$. 
The principal component scores corresponding to $G$ are given by $(G_2, \ldots, G_{p + 1})$ and hence the corresponding co-ordinates in the original high-dimensional spherical space are 

\begin{equation} X_{high} = G_1 \bar X +  \sum_{j=1}^p G_{j+1}  V_j  \in \mathcal{S}^d  \label{exactfit}
\end{equation}

If the full number of PCs is used $p=d$ then $X_{high}$ is exactly equal to the original $X_i \in \mathcal{S}^d$. 

Other transformations could be used, for example \citet{Drydenetal19} considered 
\begin{equation}\label{eq:pc.rev}
\frac{s}{\sin(s)}\left(G_2,G_3,\ldots,G_{p + 1}\right)^\top
\end{equation}
where $s = \cos^{-1}(G_1)$ for the PC scores.  Hence, an approximation in the inverse exponential map tangent space to the high-dimensional sphere is given by:
$$ X_{tan} = \bar X +  \frac{s}{\sin(s)} \sum_{j=1}^p G_{j+1}  V_j  \in \mathcal{T}_{\bar X}( \mathcal{S}^d )  . $$
A further possibility is the extrinsic approximation
$$ X_{E} = \bar X^A +  \sum_{j=1}^p G_{j+1}  V_j , $$
which is then normalized. 

\section{PNS biplot and back-fitting}
An important question of interest when carrying out principal nested spheres analysis is: what interpretation can we give to each PNS component in terms of the original variables? This is equivalent to asking: which of the original variables contribute most to each 
PNS score? When we carry out PCA we inspect the loadings for each 
component which are helpful to provide an interpretation of each PC. A large magnitude loading for a variable means that the particular variable is important, with either a positive or negative effect depending on the sign of the loading. Also, a biplot \citep{Gabriel71} is another important graphical method that is often used to display the effects of PCs. 
In the biplot, a plot of the first two PC scores is given and in addition vectors are drawn from the origin to indicate the loadings from each variable on the first two principal components.  

In this section we provide a novel type of biplot for PNS which aims to do the same task as for PCA except for 
the non-linear PNS decomposition. We call the method the {PNS biplot}. 

The PNS biplot is a useful tool for examining what effects each PNS score represents. The method involves moving along a path for 
each of the PNS scores and then transforming the path to the original space of the variables. This produces a path in the original variable space for each 
PNS score, with the x-coordinate on the biplot given by the path from PNS1 and the y-coordinate given by the path from PNS2. 
Plotting the joint x-y paths minus the PNS mean
for each variable gives the main structure of the PNS biplot. 
Each path must pass through the origin, but the 
large magnitude paths in either the x or y direction indicate variables which have a strong effect in either PNS1 or PNS2 respectively. 
In addition to the plot of the paths we also include alongside a plot of the first two PNS scores. 

This method is similar to the biplot in PCA 
\citep{Gabriel71} in
which the paths are straight line vectors, and each 
line is proportional in length in the x and 
y directions of the biplot to the PCA 
loadings of PC1 and PC2. In the usual biplot scores and vectors are given on the same diagram (hence the name ``bi" plot). However, in our high dimensional applications the plot can become very crowded, and so we provide the two plots side-by-side for clarity.

In more detail, the PNS biplot involves constructing a path for each original variable $m$ which is given by the $m$th ($m=1,\ldots,d+1$) component of the 
following vectors: 
\begin{eqnarray*} 
\mathbf{x}(\lambda_1) & = & h^{-1}\{ ( \lambda_1 , 0,\ldots, 0) | \Phi \} - \zeta_0\\
\mathbf{x}(\lambda_2) & = & h^{-1}\{ ( 0, \lambda_2 , 0,\ldots, 0) | \Phi \} - \zeta_0
\end{eqnarray*}
where $\zeta_0$ is the PNS mean and $h( )$ is the map from the original co-ordinates to the PNS scores given in Equation (\ref{gfun}). Values of $\lambda_1$ and $\lambda_2$ are both varied from $-2$ to $+2$ standard deviations and the resulting paths $(\mathbf{x}(\lambda_1), \mathbf{x}(\lambda_2))$ are plotted on an x-y plot for each original variable,  
with an arrow-head at +2 standard deviations. When  PNS score $j$ sweeps a large path for the $m$th variable it means that variable $m$ has a strong effect on the $j$th PNS score. Similarly, with a very short path the variable has little effect. Such PNS biplots are useful for interpreting the main contributions of the variables to each PNS score. We call the process of going from the PNS scores back to the original variables { back-fitting}. By plotting the first two PNS scores alongside the back-fitted path plot we can see which original variables are associated with which original observations, e.g. observations with high PNS1 scores are positively associated with the original variables that have a long path with arrow head pointing along the PNS1 axis.

Back-fitting from PNS to the original variable space is also appropriate when using fast PNS. The procedure is the same as for the PNS biplot except there is an additional stage mapping back from the 
PC scores to the original variables, using an orthogonal transformation. Using the {\tt shapes} package in R  \citep{DrydenGithub} this can be carried out using 
{\tt pns\_biplot}. 

Another related example of back-fitting is when using principal nested shape spaces (PNSS) \citep{Drydenetal19}. PNSS is similar to PNS except the data 
are geometrical configurations represented by 
$k \times m$ matrices of Cartesian co-ordinates of $k$ points in $m$ dimensions. The configurations are centered and rescaled to have unit size, and then at each stage of PNSS analysis the configurations are rotated optimally. 
In the PNSS back-fitting procedure the PNSS scores are mapped back to a geometrical configuration in the original data space which can be visualized, as demonstrated in \cite{Drydenetal19}. For example, a path $\gamma(t)$
in PNSS score space is mapped back to the original 
Cartesian coordinates, which are each functions of $t$.  Movies of 
the objects as $t$ varies can be helpful in providing interpretations to each of the PNS scores. Similar plots are available for the commonly used tangent space PCA plots in shape analysis \citep[Section 7.7]{Drydmard16}.

\section{Applications}
\subsection{Melanoma Study}
We consider a study involving $205$ patients suffering from skin cancer and operated for malignant melanoma at the German Cancer Research Center, Heidelberg \citep{Mianetal05,Browetal10}. The data were collected using SELDI mass spectrometry, where peak heights are measured for a range of mass over charge (m/z) values corresponding to peptides present in the sample. This particular dataset has $500$ m/z (mass over charge) peak heights (intensity), which are the largest peak heights are obtained using the algorithm of \citet{Browetal10}. There are $205$ patients in two groups where the first $101$ observations are Stage 1 and the last $104$ observations are Stage 4, as visualized in Figure \ref{cancer}. 

\begin{figure}[htbp]
     \centering
     \includegraphics[width=13cm]{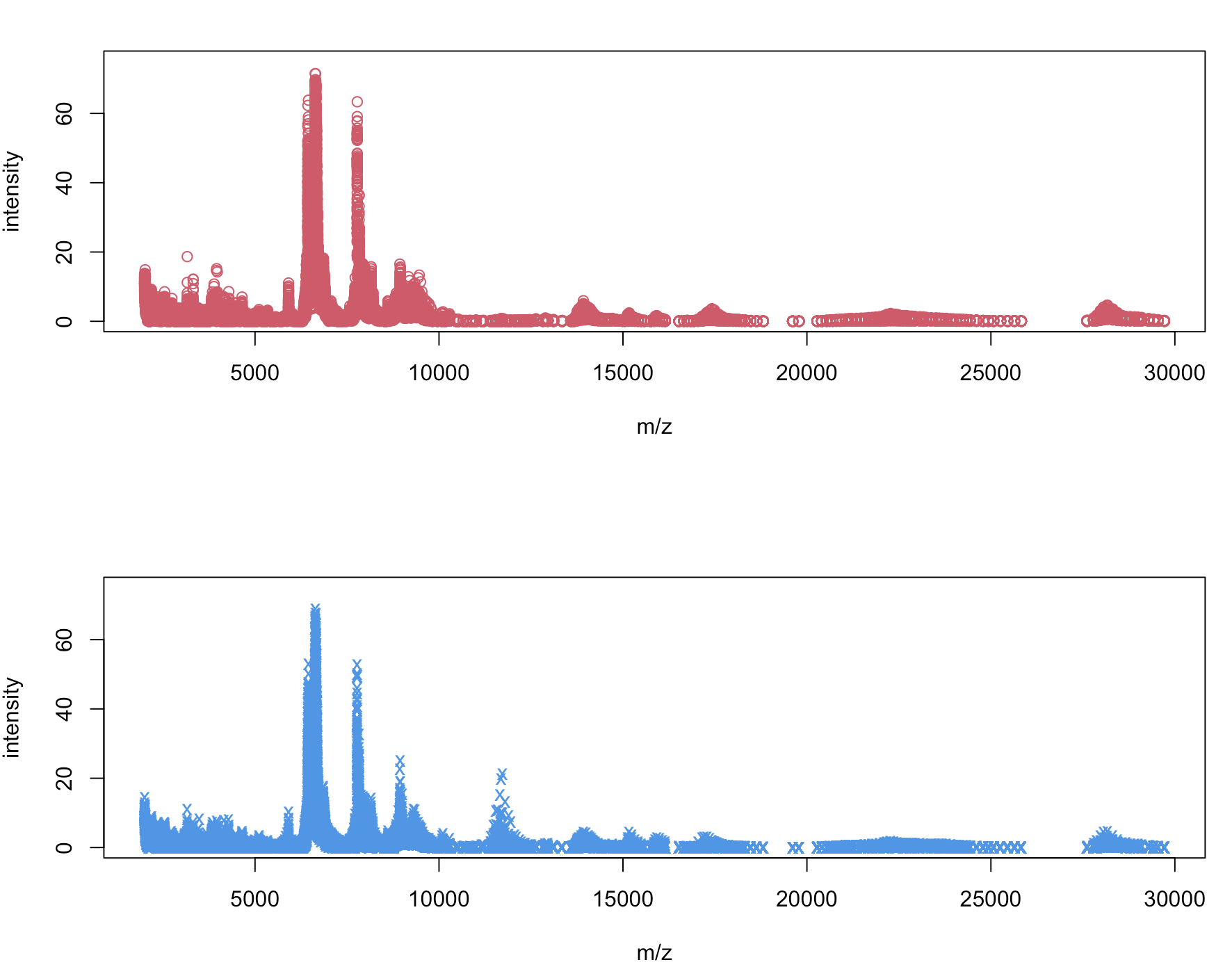}
     \caption{Melanoma proteomics data: intensity versus m/z. Stage 1 ($n=101$) patients (above, red) and Stage 4 ($n=104$) patients (below, blue).}
     \label{cancer}
\end{figure}

For such proteomics data we often wish to normalize the data, so that the overall scale of the measurements for an individual is removed. 
We normalize the data by dividing each observation vector by its Euclidean norm, so that each observation is on the sphere $S^{499}$. In order to explore lower dimensional structure we use the fast PNS method of Section \ref{fastpns} using
the R function {\tt fastpns} \citep{DrydenGithub} with $d=30$. The first 30 PCA dimensions explain $95.42\%$ of the variance, so little information is lost in using the fast PNS approximation.

We consider the different model choice methods and in 
Figure \ref{radii}(a) we see the fitted radii for each method. 
Note that all the model choice methods have very similar radii for PNS1-3 and these are much less the great subspheres radius $1$. The variance test gives the larger radii for the 
higher dimensional PNS spheres, followed by the KS test, followed by the others. In Figure \ref{radii}(b) we see the percentage of variability explained, and the first five PNS scores explain about $79-88 \%$ of the overall variability for the methods, and so we choose five PNS scores for our analysis.  

\begin{figure}[htbp]
     \centering
     \includegraphics[height=5.5cm]{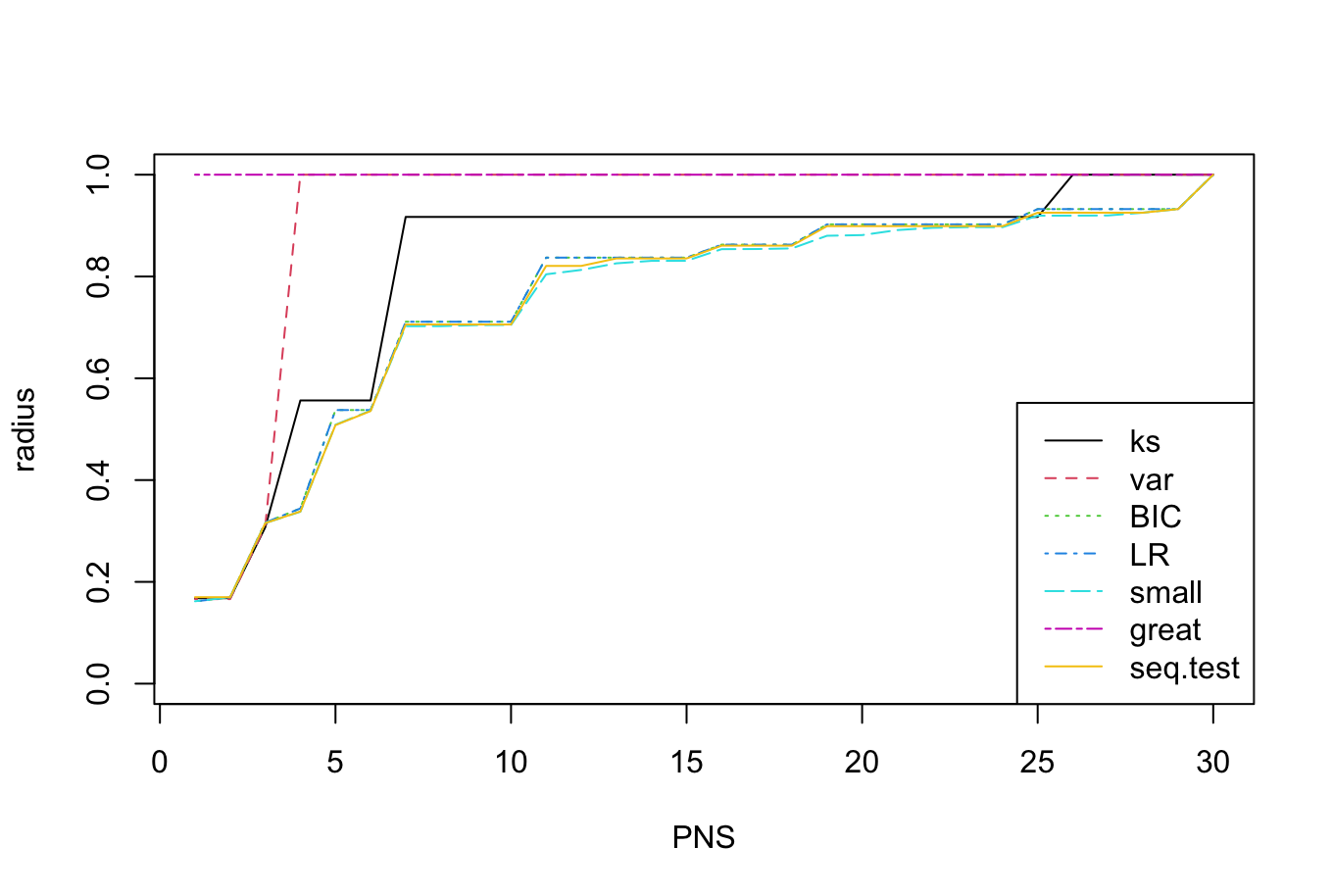}
         \includegraphics[height=5.5cm]{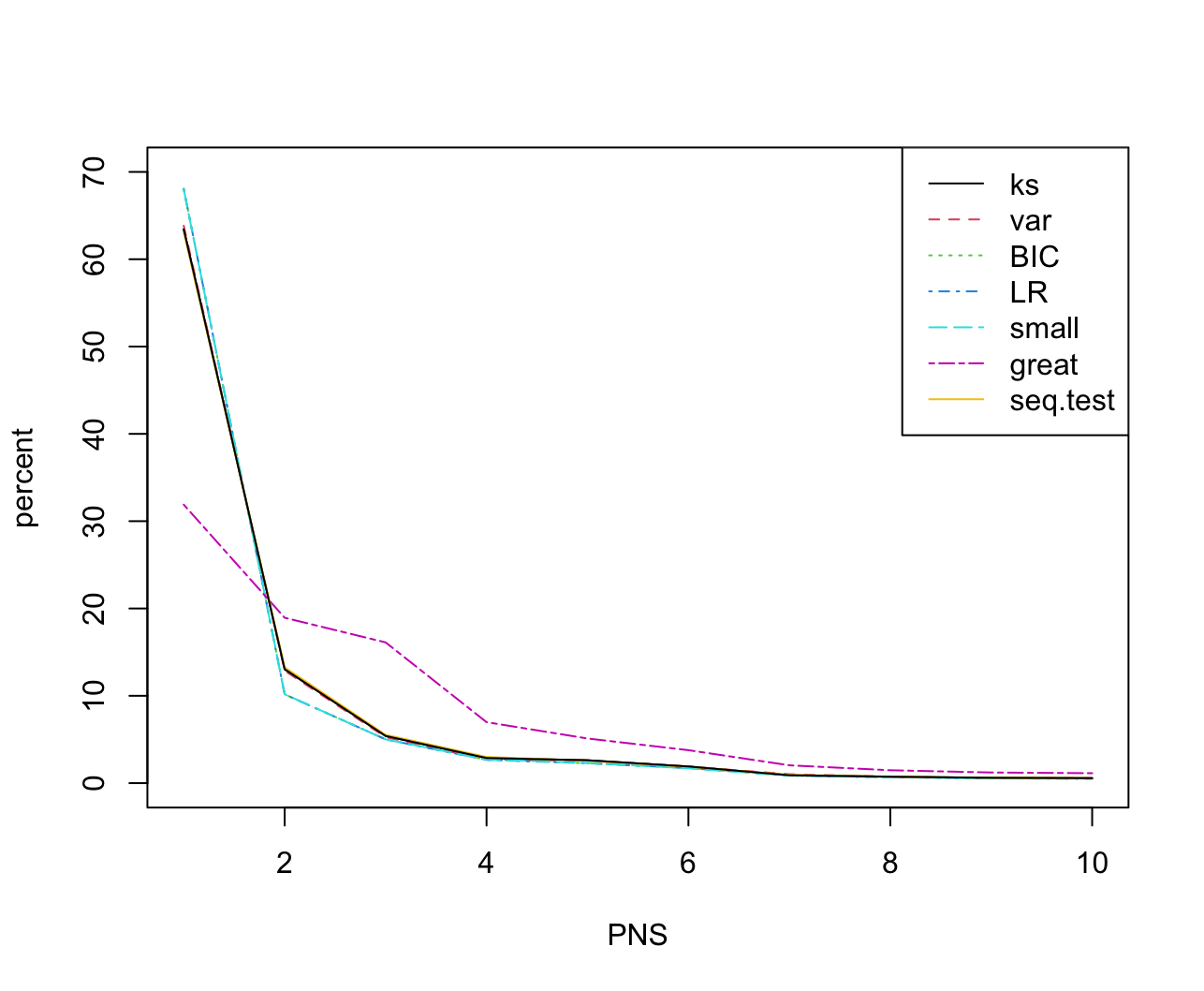}
     \caption{(a) Left: The fitted radii for each PNS level using the different model choice methods. 
     (b) Right: The percentage variability explained by the first ten PNS scores. The methods for choosing the subspheres are the Kolmogorov-Smirnov test (ks) of Section \ref{KStest}; 
     the variance test (var) of Section \ref{vartest}; BIC; likelihood ratio test (LR) of Section \ref{LR}; small subspheres (small); great subspheres (great); and the sequential test (seq.test) of \citet{Jungetal12}.}
     \label{radii}
\end{figure}

Plots of the PNS projections onto $S^2$ are given in Figure  \ref{fastpns-spheres} for the great subspheres case and the KS test case. In both cases there is a clear difference in the distribution of the Stage 1 and Stage 4 groups.

\begin{figure}[H]
    \centering
    \includegraphics[width=4.5cm]{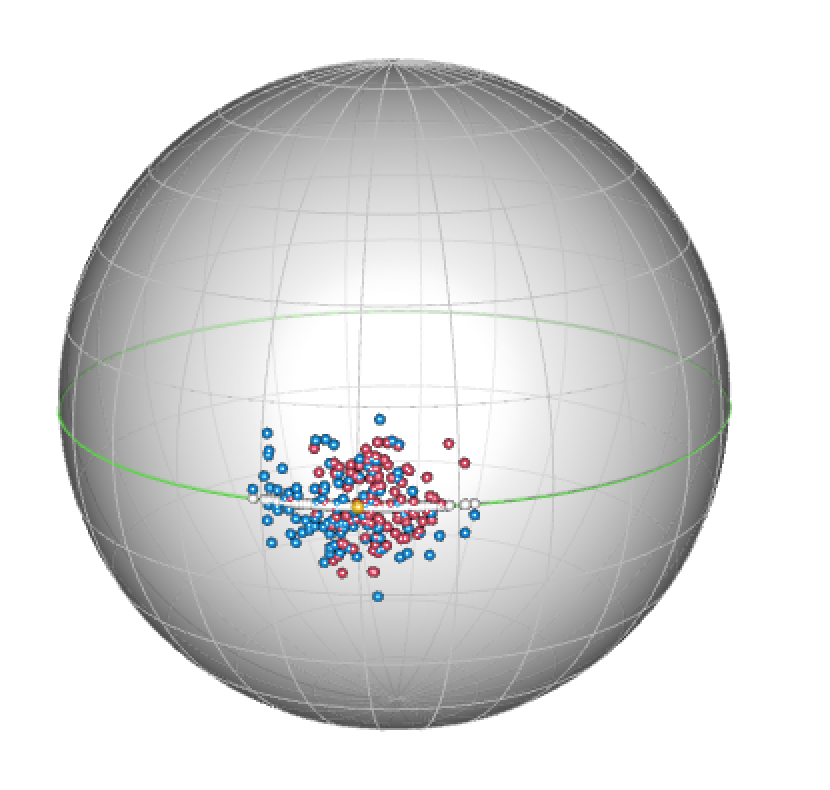}
    \includegraphics[width=4.5cm]{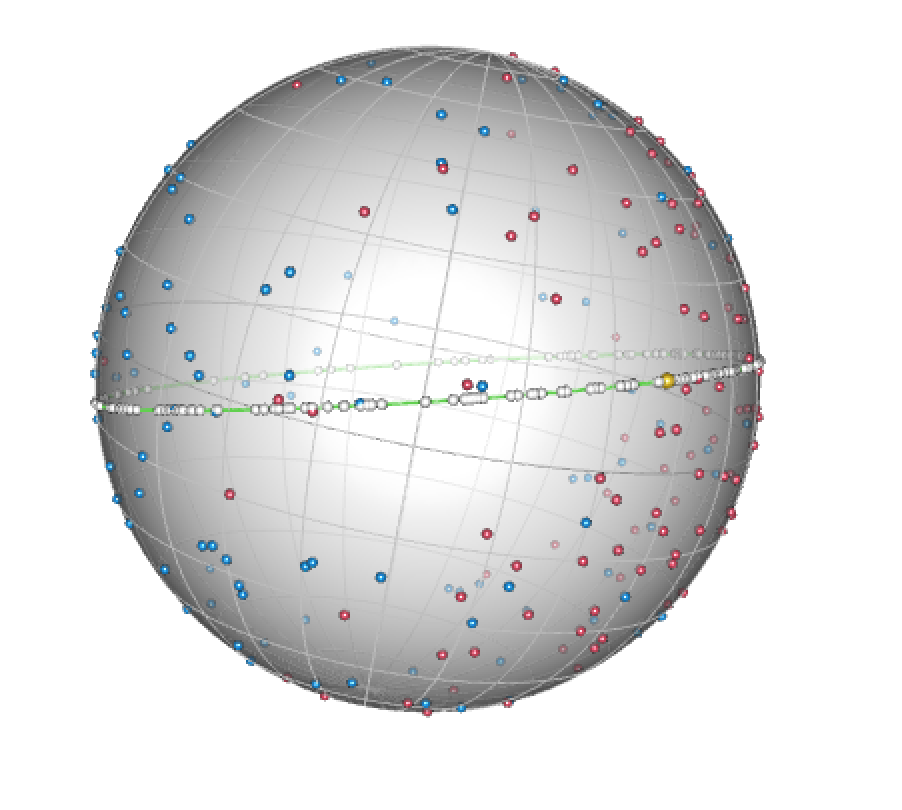}
     \includegraphics[width=7cm]{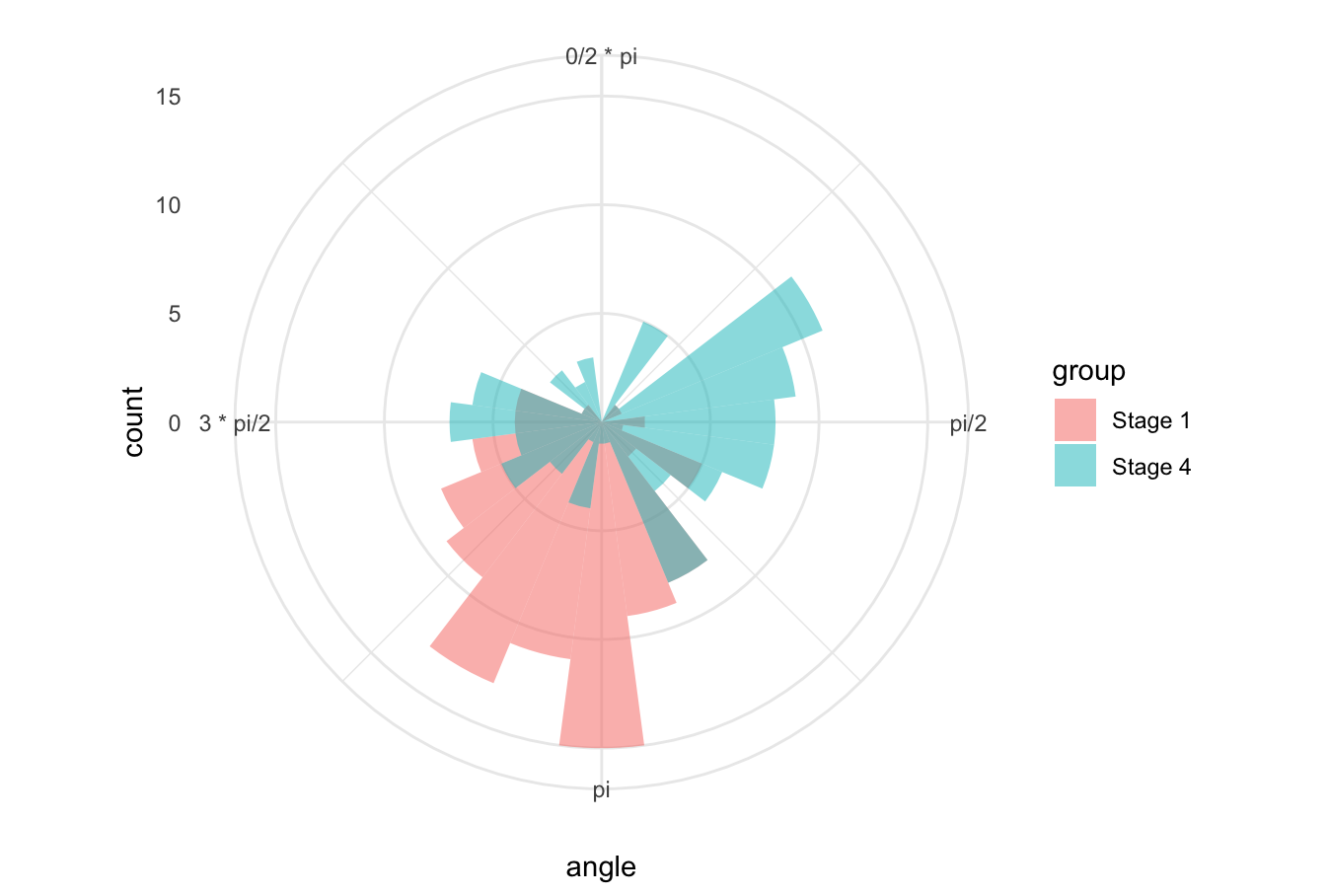}
    \caption{The fitted PNS in $S^2$ after using great subsphere fits at each stage (left) and the KS test (middle). The red points are from Stage 1 patients and blue are from Stage 4 patients. The green line is the fitted great circle, white dots are the data projected onto the fitted great circle, and the yellow dot is the PNS mean. 
    The radius of the left-hand sphere is 1 and the radius of the middle sphere is 0.168 (smaller than shown). A circular histogram for the KS PNS score 1 is shown in the right-hand plot.}
    \label{fastpns-spheres}
\end{figure}

We conducted two sample t-tests to compare great subsphere PNS scores between two groups Stage 1 ($n=101$) and Stage 4 ($n=104$) cancer patients, and we see that there are clearly significant differences in mean score between Stage 1 and 4 for PNS score 1, 3, 4 (p-value $< 0.01$), but not for PNS score 2, 5 (p-value $> 0.1$). Similar results are also obtained 
when investigating the PC scores. For the KS PNS score 1 rescaled to the unit circle we use Watson's 
two sample test for homogeneity \citep{Jammetal01} with test statistic $1.8127$ and p-value $< 0.001$. For KS PNS score 2-5  (p-value $>0.1$) we use a two sample t-test. Hence the difference between the Stage 1 and 4 distributions is only seen in KS PNS score 1, and not in KS PNS scores 2-5.

Figure \ref{fastpnsbiplot} shows the PNS biplot for this dataset using the KS PNS fit.  The effects of each variable on PNS1 and PNS2 are given by the lines with arrows for different m/z peak heights such as $6436, 6451, 6475, 6634, 6670, 7778, 7828$
which play important roles in representing the variability in the PNS1 direction, which in turn could be useful for separating the different stages given the difference in PNS1 score distributions for stages 1 and 4.  

\begin{figure}[htbp]
     \centering
     \includegraphics[width=8cm]{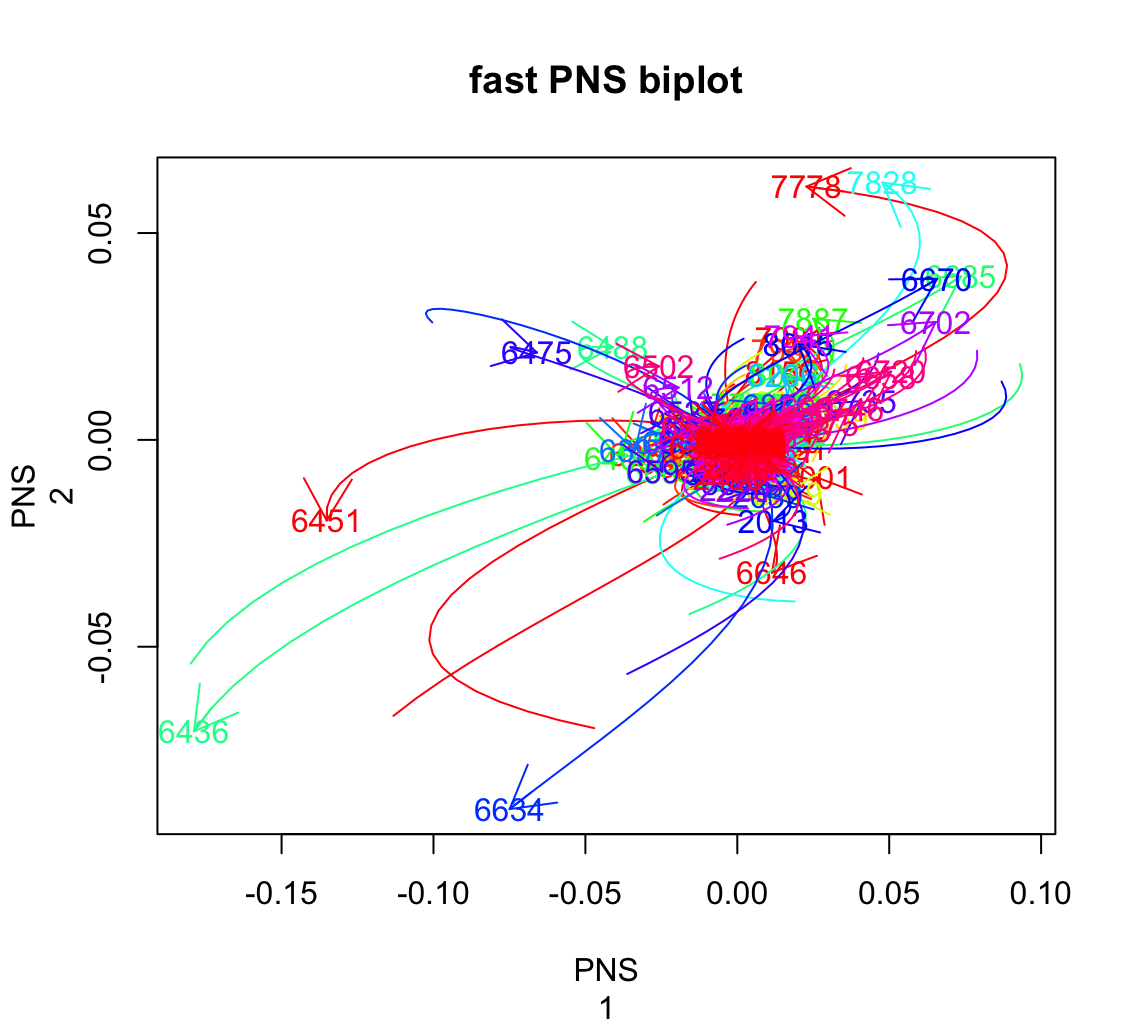}
     \includegraphics[width=8cm]{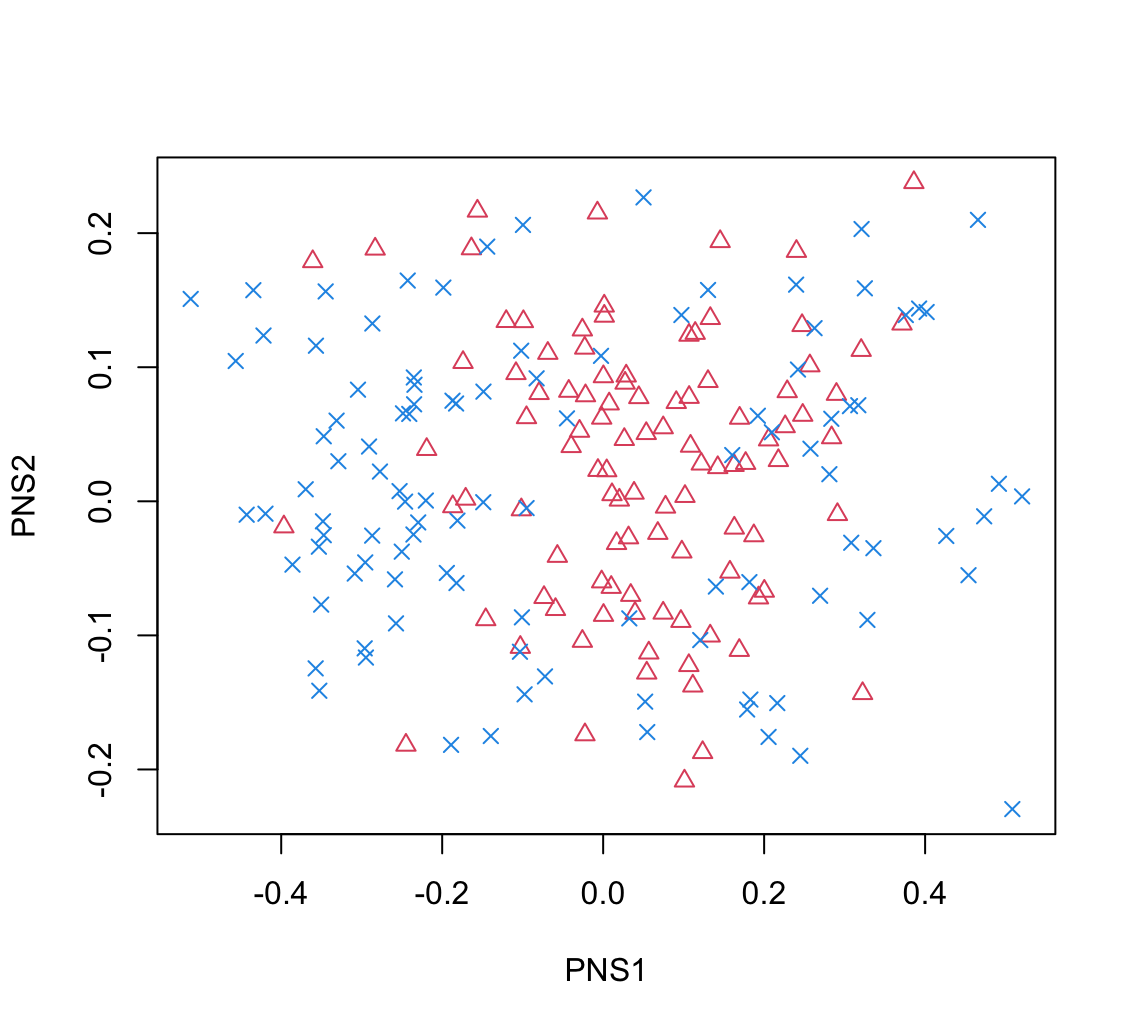}
     \caption{PNS biplot for the KS PNS fit to visualize (left) the importance of the m/z peaks for PNS1 and PNS2. The colors for the paths use a rainbow color scheme, with red being the initial ordered peaks in the dataset (larger peaks). (Right) The PNS scores for the data with Stage 1 (red triangle) and Stage 4 (blue x) indicated. }
     \label{fastpnsbiplot}
\end{figure}

We investigate the performance of logistic regression models with binary response given by Stage 1 versus Stage 4 using Monte Carlo cross validation with 90\% training data and 10\% test data from 1000 random splits: (i) 7 m/z peaks $6436, 6451, 6475, 6634, 6670, 7778, 7828$ selected from the fast PNS biplot, (ii) KS PNS scores (PNS1–PNS5), and (iii) great PNS scores (PNS1–PNS5). Model (iv) is introduced below. 

\begin{table}[ht]
\centering
\caption{Monte Carlo cross-validation results for the Melanoma data. The Accuracy is the average proportion of correct classifications in the test data; the Sensitivity is the average proportion of Stage 4 patients that are classified correctly; and the Specificity is the average proportion of Stage 1 patients classified correctly. }
\label{tab:melanoma_cv}
\begin{tabular}{l c c c}
        &       Accuracy    &    Sensitivity   &    Specificity   \\
        \toprule
(i) 7 peaks     &    0.80 & 0.77 & 0.83\\
(ii) KS PNS    &    0.75 & 0.72 & 0.80\\
(iii) great PNS  &   0.75  & 0.73 & 0.79\\
\hline
(iv) 9 peaks  &   0.85  & 0.83 & 0.87\\
\bottomrule
\end{tabular}\label{TAB1}
\end{table}

The models (i)-(iii) show good discrimination in Table \ref{TAB1} (accuracy $0.75-0.80$, sensitivity $0.72-0.77$ and specificity $0.79-0.83$), but the logistic regression model using the fast PNS biplot selected peaks is a little better overall and interpretation is easier than for the PNS models. 
For models (ii) and (iii) the PNS is estimated on the training data for each random split, and we use the sine and cosine of the PNS score 1 scaled to the unit circle as predictors, in addition to PNS scores 2-5 for prediction on the test data.

\begin{figure}[htbp]
     \centering
     \includegraphics[width=8.5cm]{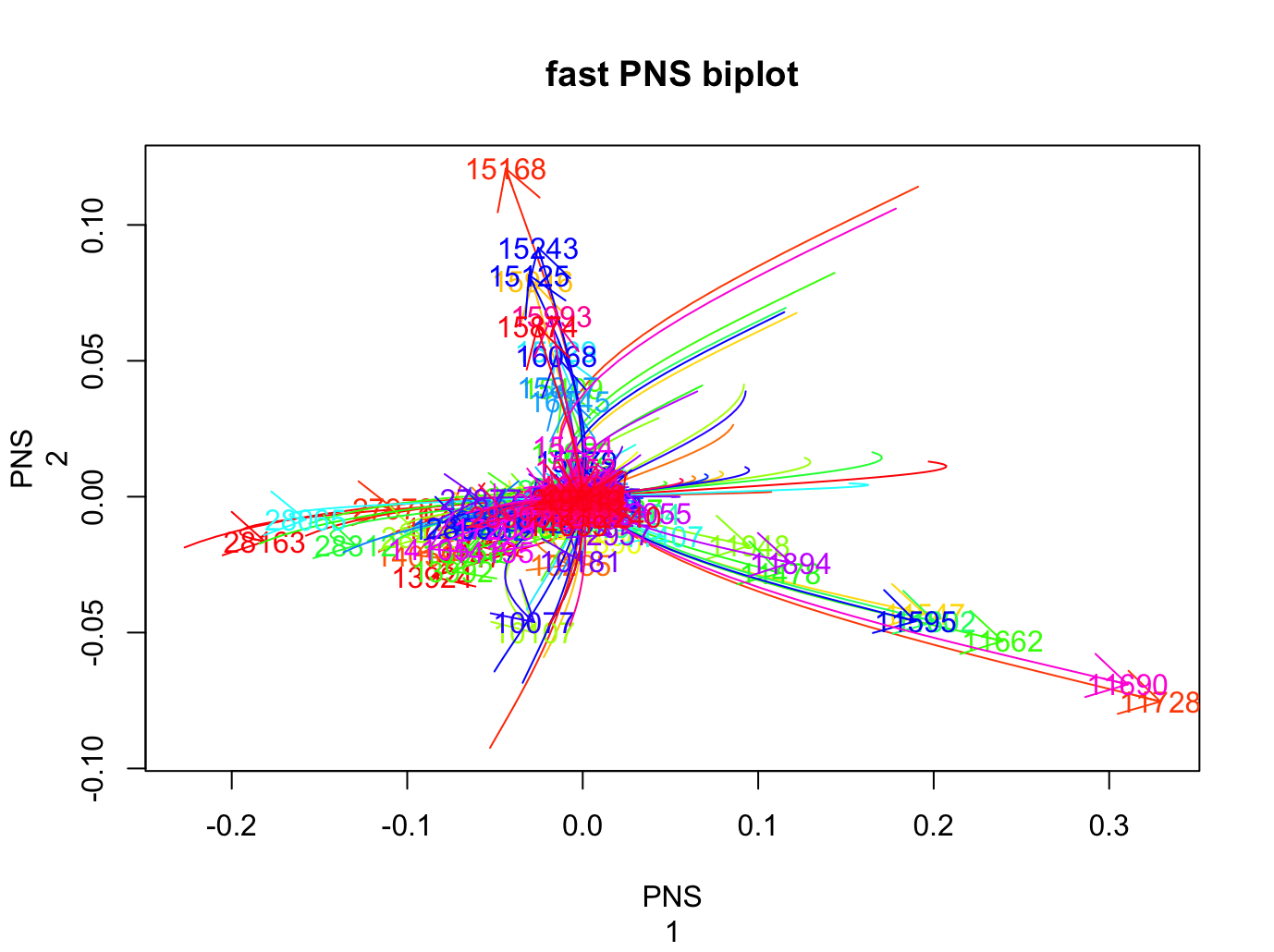}
     \includegraphics[width=7cm]{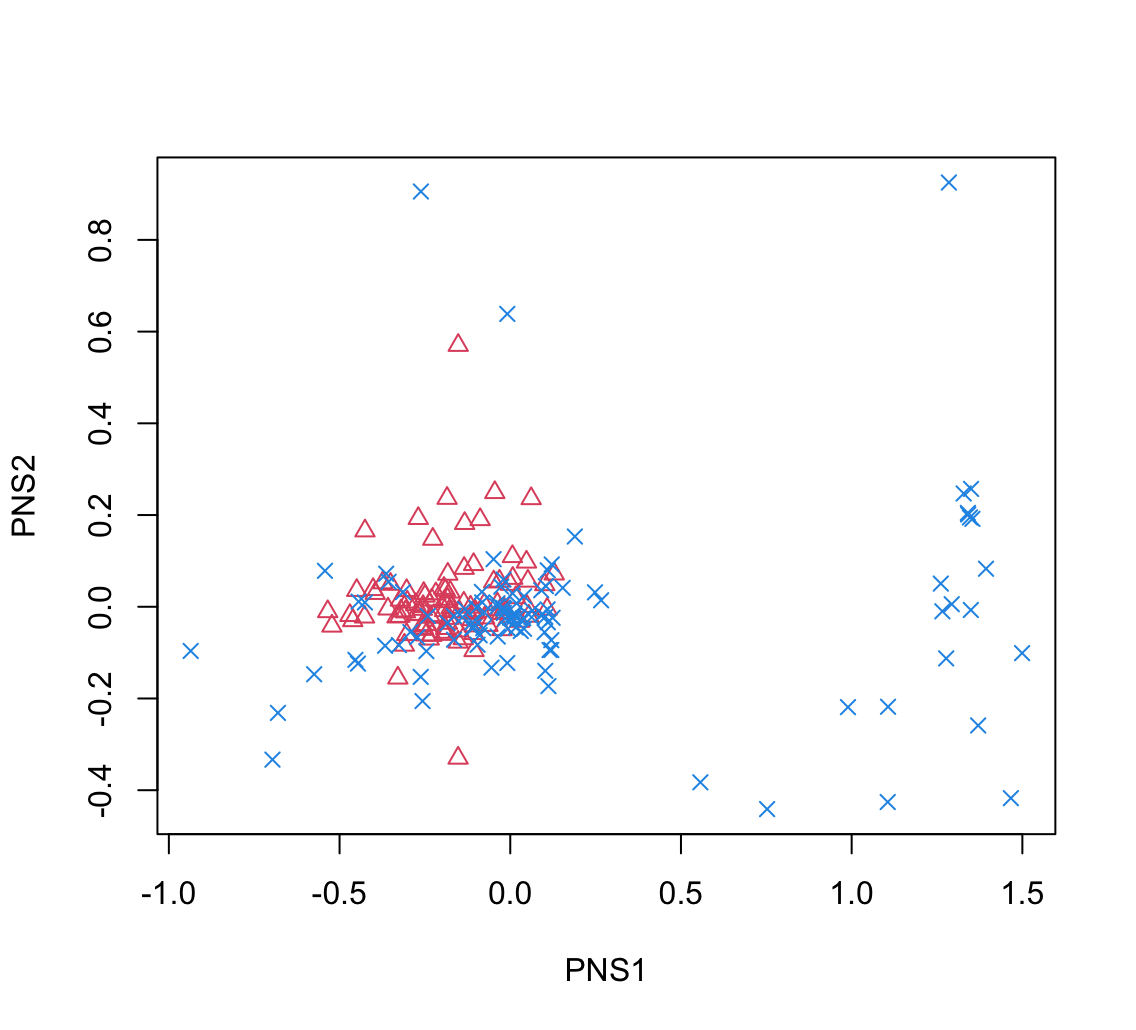}     
        \caption{PNS biplot for the KS PNS fit to visualize (left) the importance of the m/z $> 10000$ peaks for PNS1 and PNS2. The colors for the paths use a rainbow color scheme, with red being the initial ordered peaks in the dataset (larger peaks). (Right) The PNS scores for the m/z $> 10000$ data with Stage 1 (red triangle) and Stage 4 (blue x) indicated. }     
     \label{heavy}
\end{figure}

From \citet{Browetal10} some of the higher m/z values are also important for 
distinguishing between stage 1 and 4, and so we also carry out fast PNS on the 
$p=174$ peaks (normalized to the sphere $S^{173}$) which are above m/z=$10000$. In Figure \ref{heavy} we see the fast PNS biplot from KS model choice which highlights the peaks 11,728 and 28,263 as important peaks contributing to PNS1. There are other neighboring peaks that stand out too, but 11,728 and 28263 have the longest paths in the positive and negative PNS1 directions, respectively. Note that 11,728 is especially important for a sub-group of Stage 4 patients with large PNS1 scores here. 
This is an example where the PNS biplot interpretation is particularly clear: large values at m/z = 11,728 are associated with a group of Stage 4 patients with high PNS1 scores.  

We denote the logistic regression model including the two extra peaks 11,728 and 28,263 
as model (iv) in Table \ref{TAB1}, which has 9 m/z peaks in total as predictors. 
There is a clear 
improvement in classification performance. The peak at 11,728 has been identified in other studies \citep{Mianetal05,Findeisenetal09} as serum amyloid A which, together with other biomarkers, has prognostic prediction of malignant melanoma patients. The classification results using all 9 selected peaks
from the fast PNS biplots is given in Table \ref{TAB0} with $86.3\%$ accuracy, $84.6\%$ sensitivity and $88.1\%$ specificity. 

\begin{table}[h]
\centering
\caption{Classification of Melanoma data}
\begin{tabular}{l c c c}
& $\;$ & \\
& Predicted & Predicted\\
        &       Stage 1   &    Stage 4    \\
        \toprule
True Stage 1     &    89 &12 \\
True Stage 4   &    16 & 88 \\
\bottomrule
& 
\end{tabular}\label{TAB0}
\end{table}

\subsection{Pan Cancer RNA-seq data}
We consider an investigation into different types of cancer using 
RNA sequencing (RNA-seq) data. The Pan Cancer study involves 
the collection of RNA-seq data from $n=300$ patients with six different types of cancer in equal sized groups of $50$ \citep{Hoadley2014TCGA-PanCan,Marrdryd22}. The cancer groups are Bladder Cancer (BLCA), Kidney Renal Cancer (KIRC), Ovarian (OV), Head and Neck Squamous Cancer (HNSC), Colon Adenocarcinoma (COAD) and Breast Cancer (BRCA). For each patient the transcriptome is measured using RNA-seq technology to detect and quantify all RNA molecules in a sample, which in turn provides gene expression levels. In this dataset there are $d+1=12,478$ gene measurements per patient.

We normalize the data to remove scale effects by dividing each observation vector by its Euclidean norm, so that each vector is on the sphere $S^{12477}$. This is a very high-dimensional sphere and so again we shall use the fastpns method of Section \ref{fastpns}. Using the R function {\tt fastpns}, the PNS method is applied to the reduced dimension sphere based on the first $30$ principal components, i.e. a reduced sphere in $p=d+1=30+1=31$ dimensions. The fitted PNS scores using small subspheres at each stage starting with $30$ PC dimensions are shown in Figure  \ref{PNScancer}. The first $30$ PC dimensions explain $69.1\%$ of variance. So by choosing only $30$ dimensions instead of $12,478$ this saves a lot of time in computation for the analysis but at the expense of losing $30.9\%$ of the variability.

The PNS model choice methods give very similar fitted radii to each other, and we choose small subspheres fits as our main approach here, and compare it to using the great subsphere fits. Note that although small subspheres are fitted, the fitted radii are quite large, with the smallest radius 
being $0.92$, which is close to the great subsphere fits with radius $1$ at each level. 
The percentage variability explained is 
seen in the left plot in Figure \ref{PNScancer}. Note that three small PNS scores summarize 55.92\% of the overall variability. 

\begin{figure}[htbp]
    \centering
    \includegraphics[width=8cm]{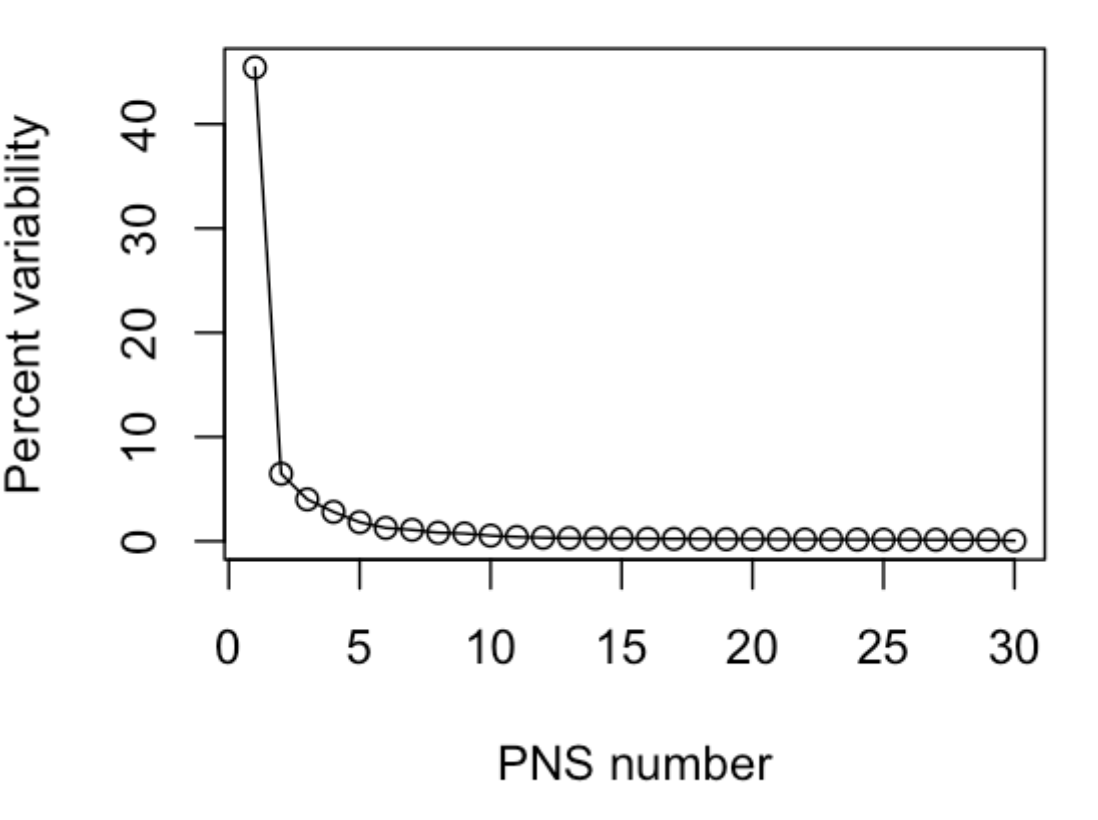}
       \includegraphics[width=7cm]{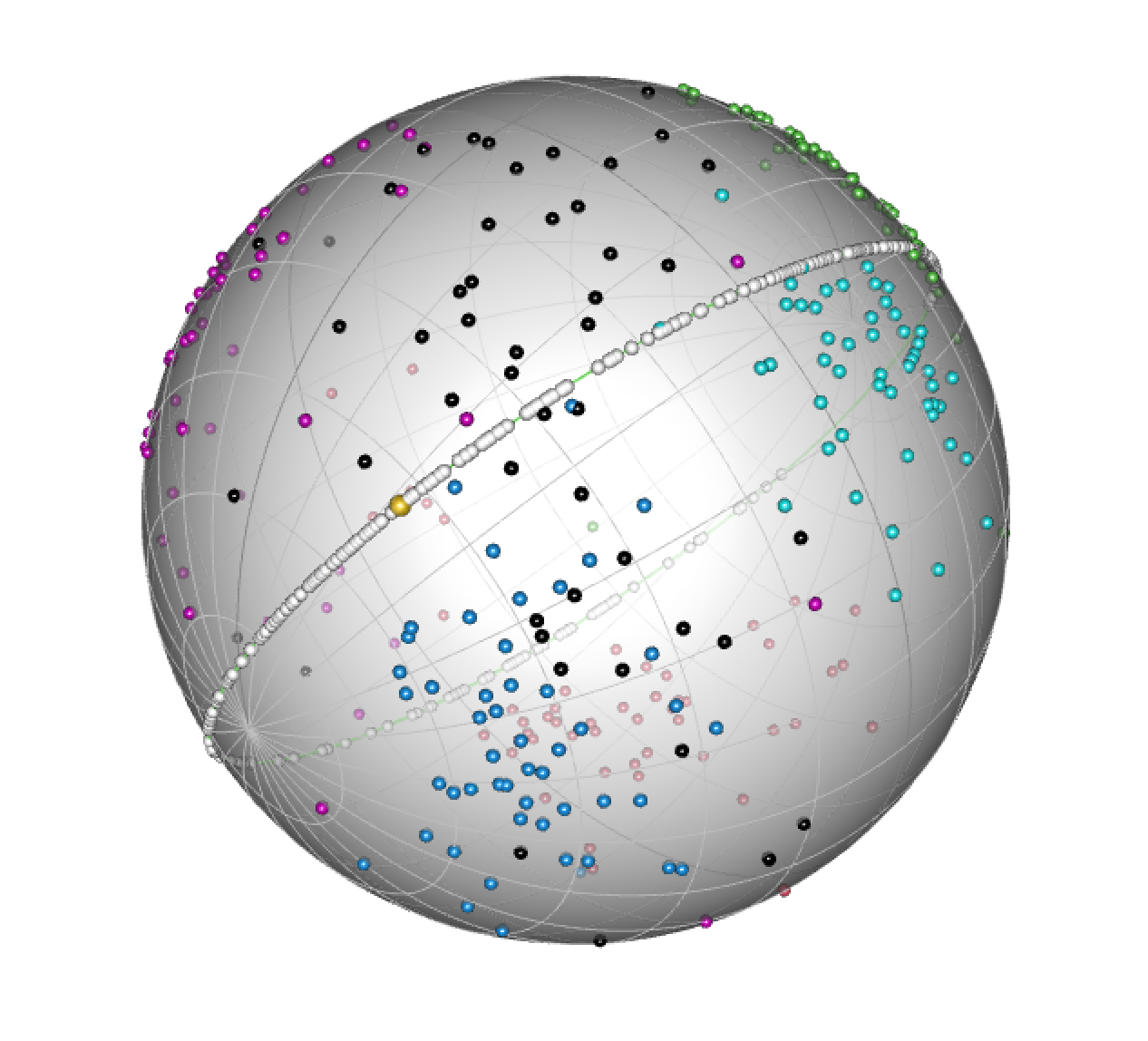}
    \caption{(Left) The percent variability explained by each PNS score for the Pan Cancer dataset, with small subsphere fits at each level. (Right) Fitted PNS with small subsphere fits at each stage. Here, different colors are for $6$ different groups. The green line is the fitted line, the white dots are projected data in the fitted line, and the yellow dot is the PNS mean of the fitted dataset.    
    }
    \label{PNScancer}
\end{figure}

A plot of the projected data on the fitted two-dimensional PNS sphere with the 
one-dimensional best fitting PNS circle is given in the right plot of Figure \ref{PNScancer}.
The projected points are spread widely over this sphere, but the distinctions between the groups can be seen.

Figure \ref{fastpnsbiplotpan} shows the PNS biplot for this dataset using the command {\tt pns\_biplot} in R \citep{DrydenGithub}. The effects of each variable on PNS1 and PNS2 are given by the curved arrows for different genes such as {\tt 
LYPD3|27076, ESR1|2099, ABP1|26, 
S100A8|6279, S100A9|6280, GATA3|2625, PAX8|7849}
which stand out from the others. There are clear distinctions between groups in the plots of PNS score 2 vs PNS score 1 
in the righthand plot of Figure \ref{fastpnsbiplotpan}. Recall that PNS score 1 is a circular variable, so those observations at about PNS1 score of -3 are actually close to those at +3.

\begin{figure}[htbp]
     \centering
     \includegraphics[width=8cm]{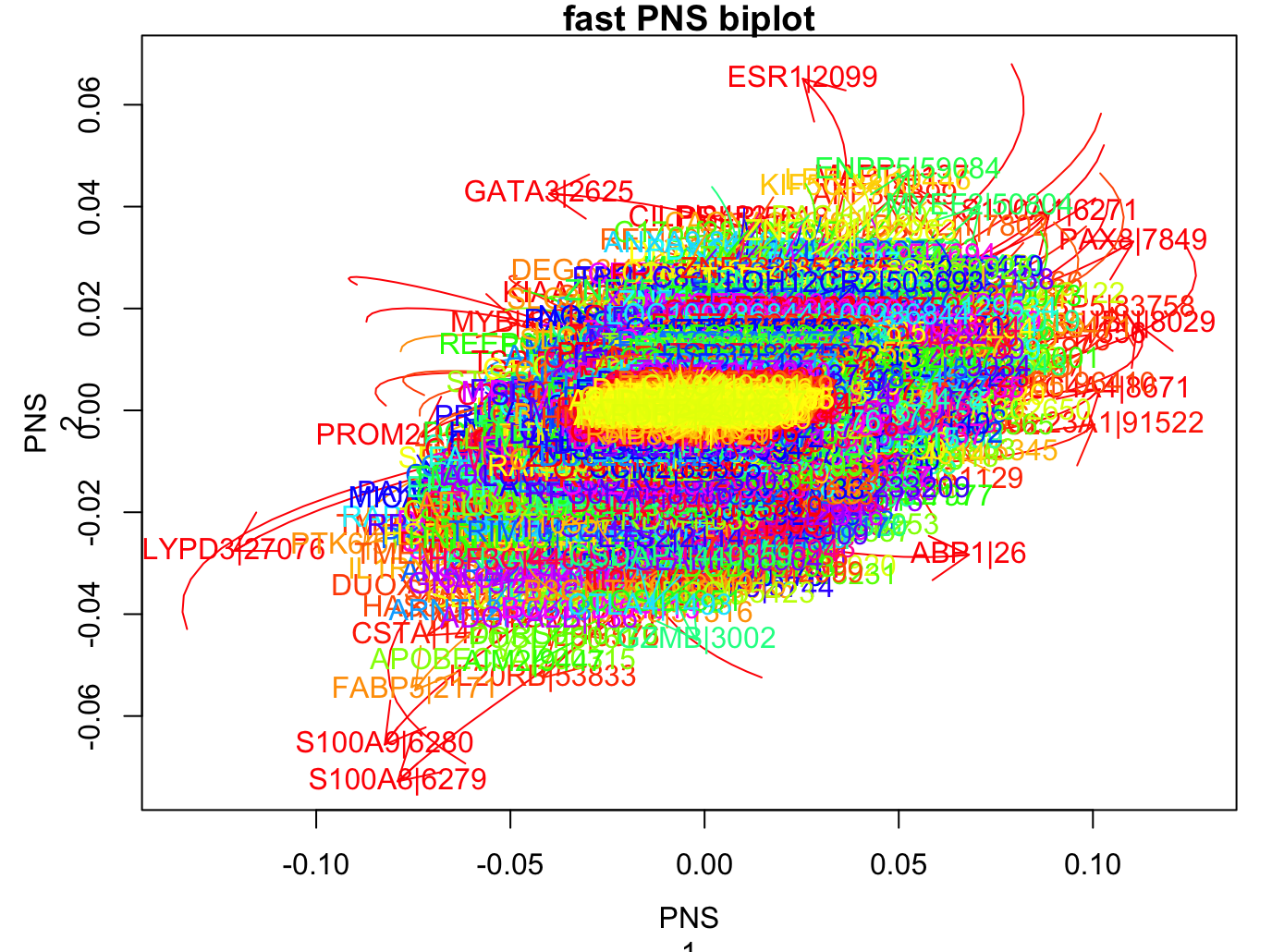}
     \includegraphics[width=8cm]{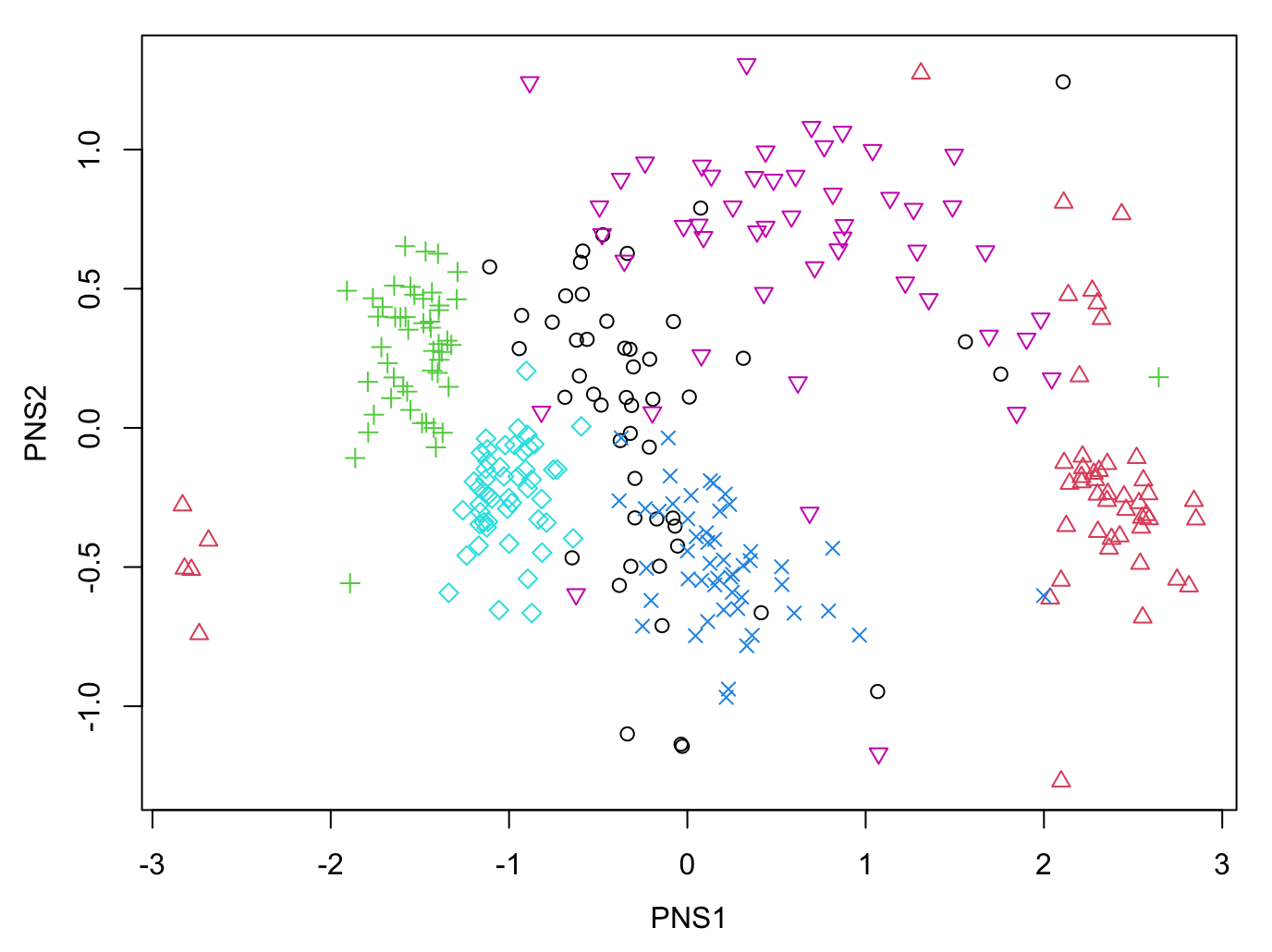}     
        \caption{PNS biplot for the small subspheres PNS fit to visualize (left) the importance of the genes for PNS1 and PNS2. 
        The colors for the paths use a rainbow color scheme, with red being the initial ordered genes in the dataset (larger intensity genes).         
        (Right) The PNS scores for the data, and
        the different colors and symbols are for different groups. The groups are: Bladder Cancer (BLCA, black circle), Kidney Renal Cancer (KIRC, red triangle), Ovarian (OV, green +), Head and Neck Squamous Cancer (HNSC, blue x), Colon Adenocarcinoma (COAD, cyan diamond), Breast Cancer (BRCA, magenta $\nabla$).}   
     \label{fastpnsbiplotpan}
\end{figure}

We use the linear discriminant analysis (LDA) to evaluate the discriminatory power of the selected genes where the outcome is six groups of cancer and the predictors are the seven standout genes from the PNS biplot. 
The overall classification accuracy is $93\%$ indicating that these genes effectively differentiate between the cancer groups, particularly 
COAD, HNSC, KIRC and OV groups which have $98\%, 96\%, 100\%, 100\%$ accuracy 
respectively.

\begin{table}[ht]
\centering
\caption{Monte Carlo cross-validation with 90\% training and 10\% test data.}
\label{tab:modelcv}
\begin{tabular}{lcc}
Model & Accuracy \\ 
\toprule
LDA (7-gene panel) & 0.92 \\ 
LDA (Small PNS scores 1--3) & 0.88 \\ 
LDA (Great PNS scores 1--3) & 0.85 \\ 
\end{tabular}\label{TAB3}
\end{table}

We evaluated three LDA models with Monte Carlo cross-validation with 1000 random splits in to training and test data. 
The predictors were (i) the 7 genes selected from the PNS biplot (LYPD3, ESR1, ABP1, S100A8, S100A9, GATA3, PAX8), 
(ii) the first three PNS scores using small subspheres, and 
(iii) the first three PNS scores using great subspheres. 
For each model, we trained on 90\% of the data selected at random (including PNS estimation if relevant) and predicted 
the groups for the 10\% held-out test data. We report the mean CV accuracy.
From Table \ref{TAB3} all models effectively classified well (Accuracy $\ge \!85\%$), with the selected gene model from the biplot performing the best. This excellent performance illustrates the effectiveness of using the PNS biplot in gene selection for prediction.

\section{Discussion}
This paper has provided a detailed summary of the PNS method, explored properties, and introduced some new techniques for high-dimensional spherical data and interpretation, notably fast PNS and the PNS biplot. 

The fast PNS method is a very practical method for analyzing high-dimensional data and could be used in other manifold dimension reduction examples. The key is to find a basis on a lower dimensional manifold which can be obtained by PCA. 
For example this technique is easily adapted for compositional data on a simplex, where 
the square root transformation of the composition maps the data onto an orthant of the sphere \citep{Scealwels11}. 

The PNS biplot is a useful visual technique for highlighting contributions of important original variables. Objective ranking of the importance of the variables can be obtained by ranking on the path length of the swept 
out biplot curves, as well as looking for distinct curves which highlight 
more independent variables. The biplot need not be restricted to the first two PNS values, for example any two PNS scores could be used. Also, plotting a third axis too would be possible in a 3D plot.  

We have introduced some new methods for deciding between great and small subspheres at each stage of the fitting procedure. The use of the KS test and variance test is quite different from existing methods, comparing the distribution of the residual lengths and the variance, respectively. There is subjectivity in which method to use, and so when the fitted radii using different methods are similar at the lower stages (as 
the last three stages for Melanoma and all stages for Pan Cancer) then this provides a set of PNS scores where results are not affected greatly by the method of model choice.

\section*{Acknowledgements}



The authors have no conflicts of interest to declare.


\bibliographystyle{apalike}
\bibliography{bibliography/PNS-highdim}
\end{document}